\documentclass[a4paper]{article}
\usepackage{amsfonts,amsmath,amsthm}
\newtheorem{theorem}{Theorem}
\newtheorem{lemma}[theorem]{Lemma}
\newtheorem{remark}[theorem]{Remark}

\usepackage{natbib}
\usepackage{url}

\title{The existence of polyhedral invariants is undecidable for linear systems}
\author{David Monniaux}

\newcommand{\bbS}{\mathbb{S}}
\newcommand{\bbZ}{\mathbb{Z}}
\newcommand{\bbQ}{\mathbb{Q}}
\newcommand{\bbR}{\mathbb{R}}
\newcommand{\bbN}{\mathbb{N}}
\newcommand{\calV}{\mathcal{V}}
\newcommand{\ve}[1]{\mathbf{#1}}

\date{September 24, 2026}

\begin{document}
\maketitle

\begin{abstract}
  The existence of polyhedral inductive invariants suitable for proving that a given control location is unreachable is undecidable for programs using only linear arithmetic over $\bbZ$ or $\bbQ$, by reduction from $2$-counter machines.
\end{abstract}

\section{Introduction}
Static analysis by abstract interpretation~\citep{CousotCousot77,Cousot78} seeks to automatically infer program invariants within a fixed \emph{abstract domain}. These invariants may prove that a \emph{safety property} is verified or, equivalently, that some bad states are unreachable.
An invariant is a property that holds at every program step, and an inductive invariant is one where this can be shown by induction on the number of program steps: it holds initially and is preserved by every transition.
Most often, and as is the case here, the invariant is decomposed into a family of invariants per control location in the program. An invariant is \emph{suitable} for proving that a control location is unreachable if it is empty at this location.

A classical example is \emph{interval analysis}: for all variables in the program at all locations, infer intervals of variation in which the variables are guaranteed to lie.
If the abstract domain is finite, or more generally if it has no infinite ascending chains, then the usual approach is to compute the strongest inductive invariant in the abstract domain as the stationary limit of ascending Kleene iterations.
However, the domain of intervals, and more generally many domains for numerical abstractions, has infinite ascending chains.
Many heuristics have been suggested to solve static analysis problems in such circumstances: widening operators, narrowing iterations, thresholds extracted from the syntax of the program, guided iterations and so on. None of these guarantees finding the strongest inductive invariant in the domain.
These approaches are heuristic: they may fail to find an inductive invariant suitable for proving the desired safety property even if one exists in the abstract domain.

For certain abstract domains, certain safety properties and certain classes of instructions in the program, the existence of suitable inductive invariants is known to be decidable.
Consider the case where the domain is defined by a polynomial template (a formula built from conjunction, disjunction, and polynomial equalities and inequalities over program variables and some parameters), the safety condition is defined using polynomials, the program consists only of polynomial conditions and nondeterministic assignments, and the program uses variables in $\bbR$---interval analysis is the case where the template consists of inequalities $l_x \leq x \leq u_x$ for every variable~$x$.
Then, the inductiveness condition and the safety property can be expressed as a constraint on the parameters of the template in the theory of real closed fields, which is decidable by quantifier elimination.
Similarly, if everything is affine linear, if the parameters occur linearly in the template, and the program uses variables in $\bbQ$, this boils down to the theory of \emph{linear real arithmetic}, and if over $\bbZ$, to Presburger arithmetic (also known as \emph{linear integer arithmetic}), which are likewise decidable.%
\footnote{These quantifier elimination methods do not scale. Other methods, such as \emph{policy iteration}, are preferred for practical purposes.}
This includes in particular \emph{template polyhedra} (fixed $A$, parameters $b$, system of inequalities $Ax \leq b$), including \emph{octagons} (systems of $\pm x \leq b$, $\pm x \pm y \leq b$) and intervals.

The domain of convex polyhedra \citep{CousotHalbwachs78,Halbwachs_PhD} is of particular interest. A convex polyhedron is defined by a finite system of affine linear inequalities and equalities or, equivalently, a finite system of generators (vertices and, if unbounded, rays and lines)~\citep{Schrijver98}.
Note that, contrary to templates, we do not bound in advance the number of inequalities.
There exist several libraries (NewPolka / Apron, PPL, VPL\dots) implementing this domain for static analysis, including several variants of widening operators. Again, these approaches are heuristic, and fragile---for instance, the analysis may succeed in proving properties on a program, but just adding some observer computation to that program (a computation that does not feed back into the program and does not influence the safety property) may cause the analysis to fail~\citep{DBLP:journals/entcs/MonniauxG12}.
The question whether these heuristics are actually necessary because of undecidability was thus raised~\citep{Monniaux_Acta_Informatica_2018}.

It is known that, for a system of transitions defined using linear integer or rational arithmetic plus a quadratic guard, the existence of polyhedral inductive invariants suitable for proving the unreachability of a location is undecidable~\citep{Monniaux_Acta_Informatica_2018}.
\cite{DBLP:journals/jacm/FijalkowLOOPW25} proved undecidability for unguarded affine linear programs with semilinear sets as inductive invariants, and decidability for closed semilinear invariants and simple affine loops; the original problem over polyhedra has so far remained open.

In this paper, we show that the existence of polyhedral invariants is undecidable for programs with integer or rational variables and linear arithmetic only.
We assume general knowledge of results on convex polyhedra, and refer the reader to, e.g., \cite{Schrijver98}.

\section{Reduction from linear counter machines}
Let $\bbS$ be $\bbZ$ or $\bbQ$, $m \geq 1$ and $Q$ a finite control location space.
We will reduce the problem of the termination of a deterministic linear $m$-counter machine over $\bbS$ with control locations in $Q$ to that of finding inductive polyhedral invariants for a transition system with control locations in $Q$ and $m+4$ coordinates in~$\bbS$.

Let $q_{\text{start}}$ in $Q$ be the starting state.
The transitions of the counter machine are defined by a function $d: Q \times \bbS^m \rightarrow Q \cup \{\bot\}$ and a partial function $c: Q \times \bbS^m \rightarrow \bbS^m$, which need only be defined for $(q,\ve{x})$ such that $d(q,\ve{x})\neq\bot$. The symbol $\bot$ (not belonging to $Q$) indicates that execution ends.
The \emph{execution run} of the machine is defined by $q_0=q_{\text{start}}$, $\ve{x}_0=\ve{0}$, $q_{n+1}=d(q_n,\ve{x}_n)$ (if $d(q_n,\ve{x}_n) \neq \bot$) and $\ve{x}_{n+1}=c(q_n,\ve{x}_n)$. The run is finite, of length $n+1$, if $d(q_n,\ve{x}_n)=\bot$; if no such $n$ exists then it is infinite.

Assume the counter machine to be \emph{linear}: the sets $\{ \ve{x} \mid d(q,\ve{x})=q' \}$ are finite unions of polyhedra over $\bbS^m$ (for $\bbS=\bbZ$, a polyhedron over $\bbZ^m$ is the set of integer points of a rational polyhedron), and the functions $\ve{x} \mapsto c(q,\ve{x})$ are piecewise affine (with finite decomposition).

To this machine, we associate a deterministic transition system over $(Q \cup \{q_\text{bad}\}) \times \bbS^{m+4}$, with initial state $(q_{\text{start}},\ve{0},0,0,0,0)$ (we split the $m+4$ numerical coordinates into a vector of dimension $m$ and $4$ extra coordinates) and three kinds of transitions:
\begin{itemize}
\item \emph{machine} transitions $(q,\ve{x},t,y,0,0) \mapsto (d(q,\ve{x}),c(q,\ve{x}),t+1,y+t+1,t+1,y+t+1)$, where $q \in Q$ and $d(q,\ve{x}) \neq \bot$, mimic those of the counter machine;
\item \emph{rewinding} transitions $(q,\ve{x},t,y,t',y') \mapsto (q,\ve{x},t,y,t'-1,y'-t')$, where $q \in Q$ and $t' \geq 1$, rewind the last two coordinates along a parabola;
\item \emph{bad} transitions: any state $(q,\ve{x},t,y,0,y')$ where $q \in Q$ and $y' \leq -1$ maps to
  \[ (q_{\text{bad}},\ve{0},0,0,0,0). \]
\end{itemize}

Note how, as in \cite{Monniaux_Acta_Informatica_2018}, our encoding forces the concrete execution to step along parabolas.
Let $\theta(t)=\frac{t(t+1)}{2}$. One steps along a parabola $y-\theta(t)=\text{constant}$ using $(t,y) \mapsto (t+1,y+t+1)$ and $(t,y) \mapsto (t-1,y-t)$.

A state $(q,\ve{x},t,\theta(t),0,0)$, $t \in \bbN$, mimics the state at time $t$ of the counter machine.
The execution of the transition system mimics the execution of the counter machine as follows.
The transition of the counter machine at time $t$ is replaced by, in succession
\begin{itemize}
\item the machine transition
  \[ (q,\ve{x},t,\theta(t),0,0) \rightarrow
     \left(d(q,\ve{x}),c(q,\ve{x}),t+1,\theta(t+1),t+1,\theta(t+1)\right); \]
\item $t+1$ rewinding transitions
  \begin{align*}
    & \left(d(q,\ve{x}),c(q,\ve{x}),t+1,\theta(t+1),t+1,\theta(t+1)\right) \\
    \rightarrow{} & \left(d(q,\ve{x}),c(q,\ve{x}),t+1,\theta(t+1),t,\theta(t)\right) \\
    \rightarrow{} & \dots \\
    \rightarrow{} & \left(d(q,\ve{x}),c(q,\ve{x}),t+1,\theta(t+1),0,\theta(0)\right)
  \end{align*}
  (recall that $\theta(0)=0$).
\end{itemize}
The transition system terminates if and only if the counter machine terminates.
No bad transition is exercised during this run.

The difference with the reduction presented in \cite{Monniaux_Acta_Informatica_2018} is as follows. The earlier reduction also used a parabola with the execution time as abscissa, but in order to get rid of ``mixture states'' induced by convex hulls, it used the quadratic guard $y=\theta(t)$ to restrict transitions to the true states, which lie on the parabola. The rewinding gadget allows us to do away with the quadratic guard. The idea is that taking a mixture of states on a parabola $y=\theta(t)$ gives a state strictly above a parabola: a positive \emph{defect} $y - \theta(t)$ can never go away, and the machine transitions filter for zero defect.

Let $R_K$ be the bounded closed convex polyhedron generated by the vertices $\left(k,\theta(k)\right)_{k=0}^K$ and $(0,\theta(K))$ or, equivalently, the set of points above the polygonal line defined by $\left(k,\theta(k)\right)_{k=0}^K$ and with $y$-coordinate below $\theta(K)$.
Define the family of linear functions $\sigma_i(t',y')=2y'-(2i+1)t'+i^2$. Note that
\begin{itemize}
\item $y=\frac{1}{2}\left((2i+1)t-i^2\right)$ is the equation of the tangent to the $y=\theta(t)$ parabola at $t=i$, thus, with a $2$ scaling factor, $\sigma_i$ measures the gap between $y'$ and the tangent;
\item $\sigma_i(t',y') = 2(y'-\theta(t')) + (t'-i)^2$;
\item $\sigma_i(n,\theta(n))=(n-i)^2$, thus, for $0 \leq i \leq K$, $\sigma_i$ is nonnegative on $R_K$ and vanishes only at $(i,\theta(i))$---it measures a form of ``slack'';
\item $\sigma_i(t'-1, y'-t') = \sigma_{i+1}(t', y')$: one rewinding step turns the slack at index $i+1$ into the slack at index~$i$;
\item if $y' > \theta(t')$ (positive defect), then $\sigma_i(t',y') > 0$ for all~$i$.
\end{itemize}

\begin{lemma}\label{lem:closed_termination_implies_invariants}
  Assume that for all $q,q' \in Q$, $\{ \ve{x} \mid d(q,\ve{x})=q' \}$ is a finite union of \emph{closed} convex polyhedra and that the piecewise decomposition of each $c(q,\cdot)$ is according to closed convex polyhedra.
  Assume the counter machine, thus the transition system, terminates.
  Then the transition system admits bounded polyhedral inductive invariants $I_q$ such that $I_{q_\text{bad}} = \emptyset$.
\end{lemma}

\begin{proof}
  Let $(q_t,\ve{x}_t)_{0 \leq t \leq N}$ be the run of the counter machine, where $N$ is the time at which it terminates: $d(q_N,\ve{x}_N)=\bot$.
  Let $G_q$ be the convex hull of the points $(\ve{x}_t,t,\theta(t))$ such that $q_t=q$.
  If $q$ is visited by the run, $G_q$ is nonempty, and this bounded convex polyhedron in $\bbQ^{m+2}$ can be defined as the intersection of a finite family of half-spaces $\{ \ve{z} \mid \ve{g}_{q,j} \cdot \ve{z} \leq b_{q,j} \}$, for $1 \leq j \leq H_q$. The recession cone of such a description is $\{0\}$.
  If $q$ is not visited, $G_q$ is empty, and we use
  $\pm z_k \leq -1$ for every coordinate $k$ as a description,
  which also has $\{0\}$ as recession cone.

  Let $W$ be a constant, expressing a kind of ``fattening factor''. We will collect along the proof what requirements it must fulfill. Define for all $q \in Q$
  \begin{equation}
    I_q = \{ (\ve{z}, t', y') \mid (t',y') \in R_{N+1} \land
    \bigwedge_{0 \leq i \leq N+1} \bigwedge_{j=1}^{H_q}
    \ve{g}_{q,j} \cdot \ve{z} \leq b_{q,j} + W \sigma_i(t', y') \}
  \end{equation}
  and also $I_{q_{\text{bad}}}=\emptyset$.
  The $I_q$ are defined by finitely many nonstrict inequalities and are thus closed and convex.
  The $t'$ and $y'$ coordinates are bounded due to $R_{N+1}$ being bounded, and so is $\sigma_i(t', y')$.
  The $\ve{z}$ coordinate thus lies in a polyhedron defined by a system of inequalities with the same recession cone as the description of $G_q$, which is $\{0\}$, thus is bounded.
  
  We intend that the $I_q \cap \bbS^{m+4}$ form inductive invariants; what constraints does this bring on $W$?
  A sufficient condition for the initial state $(q_{\text{start}},\ve{0},0,0,0,0)$ to be inside $I_{q_{\text{start}}}$ is that $W \geq 0$.

  Let us consider machine transitions.
  Let  $(\ve{z}, 0, 0) \in I_q$.
  Since $\sigma_0(0,0) = 0$, we have, for all $j$,
  $\ve{g}_{q,j} \cdot \ve{z} \leq b_{q,j}$;
  but that is the definition of the polyhedron $G_q$ as intersection of half-spaces.
  Using its definition as a convex hull, there exist weights
  $\lambda_n \geq 0$ for $n=0 \dots N$ such that $\sum_n \lambda_n=1$
  and
  $\ve{z} = \sum_{n=0}^N \lambda_n (\ve{x}_n, n, \theta(n))$.
  Projecting on the last two coordinates,
  $(t,y) = \sum_{n=0}^N \lambda_n (n, \theta(n))$.
  Because of the ``right stepping'' property of the parabola,
  $(t+1,y+t+1) = \sum_{n=1}^{N+1} \lambda_{n-1} (n, \theta(n))$.
  By definition of $R_{N+1}$, this implies $(t+1,y+t+1) \in R_{N+1}$.

  Let us now deal with the second part of the invariant.
  Let $q' \in Q$.
  Decompose $\ve{u} \mapsto c(q,\ve{u})$, where $d(q,\ve{u})=q'$, into a finite number of affine pieces, each with a domain of definition along a closed convex polyhedron---\emph{this is where the restriction to affine linear transitions with closed convex guards matters}.
  We are only interested in $\ve{u}$ lying in the projection of $G_q$, which is bounded.
  The function decomposition can thus be done using closed bounded components over $\bbQ^{m+2}$.
  We collect the finitely many vertices of these components into a family~$\calV$.
  Let $0 \leq i \leq N+1$, $1 \leq j \leq H_{q'}$, the linear inequality
  \begin{equation}\label{ineq:slack_ineq}
  \ve{g}_{q',j} \cdot \ve{z'} \leq b_{q',j} + W \sigma_i(t+1, y+t+1)
  \end{equation}
  is true for $\ve{z'}=(c(q,\ve{x}_{\text{from}}),t+1,y+t+1)$ such that $q \in Q$, $(\ve{x}_{\text{from}},t,y) \in G_q$, $d(q,\ve{x}_{\text{from}})=q'$,
  if it is true for all $(\ve{x}_{\text{from}},t,y)$ in~$\calV$.

  Let us distinguish in the family $\calV$ two cases
  \begin{itemize}
  \item $\sigma_i(t+1, y+t+1) = 0$, which can happen only
    if  $i=t+1$, $y+t+1=\theta(t+1)$;
    thus $y=\theta(t)$;
    but in $G_q$ the only point of the form $(\ve{x}_{\text{from}},t,\theta(t))$, if any, is for
    $\ve{x}_{\text{from}}=\ve{x}_t$ and $q=q_t$;
    then  $d(q,\ve{x}_{\text{from}})=q'$ by the definition of $\calV$,
    hence $\ve{z'}$ is $(\ve{x}_{t+1},t+1,\theta(t+1))$, and
    the inequality~\eqref{ineq:slack_ineq} becomes
    $\ve{g}_{q',j} \cdot (\ve{x}_{t+1},t+1,\theta(t+1)) \leq b_{q',j}$;
    but that is true by definition of~$G_{q'}$ as the intersection of half-spaces, since $(\ve{x}_{t+1},t+1,\theta(t+1))$ is one of the generators of~$G_{q'}$.
    
  \item $\sigma_i(t+1, y+t+1) > 0$;
    then we collect the constraint
    \begin{equation}
      W \geq \frac{\ve{g}_{q',j} \cdot (c(q,\ve{x}_{\text{from}}),t+1,y+t+1) - b_{q',j}}{\sigma_i(t+1, y+t+1)}
    \end{equation}
    which is sufficient for inequality~\eqref{ineq:slack_ineq} to hold.
  \end{itemize}
  
  Let us now consider rewinding transitions.
  Let $(\ve{z}, t', y') \in I_q$, stepping to $(\ve{z}, t'-1, y'-t')$. Since $(t',y')$ is in $R_{N+1}$ and $t' \geq 1$, $(\ve{z}, t'-1, y'-t')$ also is in~$R_{N+1}$: the $y'$ coordinate can only decrease (so remains at most $\theta(N+1)$), and the distance above the polygonal line that forms the bottom of $R_{N+1}$ remains constant.
  
  Let $1 \leq j \leq H_q$.
  For any $0 \leq i \leq N+1$, we have
  $\ve{g}_{q,j} \cdot \ve{z} \leq b_{q,j} + W \sigma_i(t', y')$.
  Because $\sigma_i(t'-1, y'-t') = \sigma_{i+1}(t', y')$, we
  have, for all $0 \leq i \leq N$,
  \begin{equation}\label{neq:stable_rewind}
  \ve{g}_{q,j} \cdot \ve{z} \leq b_{q,j} + W \sigma_i(t'-1, y'-t')
  \end{equation}
  Let us now consider the case $i=N+1$.
  $\sigma_{i+1}(t',y') - \sigma_i(t',y') = 2(i-t')+1$.
  Since $(t',y') \in R_{N+1}$, $t' \leq i$ and thus $\sigma_{i+1}(t',y') - \sigma_i(t',y') \geq 0$.
  Assuming, again, $W \geq 0$, we obtain
  $\ve{g}_{q,j} \cdot \ve{z} \leq b_{q,j} + W \sigma_i(t', y')
  \leq b_{q,j} + W \sigma_{i+1}(t', y')$.
  Since $\sigma_{i+1}(t', y') = \sigma_i(t'-1,y'-t')$,
  the inequality~\eqref{neq:stable_rewind} also holds for $i=N+1$.
  Thus, $W \geq 0$ is a sufficient condition for $I_q$ to be invariant under rewinding.

  We have thus accumulated a nonempty finite family of constraints over $W$ of the form $W \geq \textit{constant}$, we can pick $W$ as the maximum of these constants.

  Finally, let us show that these invariants imply that $q_{\text{bad}}$ is unreachable.
  Assume $q \in Q$ and a transition from a state $(\ve{z}, 0, y') \in I_q$ with $y' \leq -1$.
  By the definition of $I_q$, $(0, y') \in R_{N+1}$;
  and by the definition of $R_{N+1}$ as a convex hull of vertices with nonnegative $y$-coordinates, $y' \geq 0$.
  Such a transition is thus impossible.
\end{proof}

\begin{lemma}\label{lem:nontermination_implies_no_invariants}
  If the execution run of the counter machine is infinite, then no family of polyhedral inductive invariants exists that are suitable for proving the unreachability of the bad state in the transition system.
\end{lemma}

\begin{proof}
  Assume the existence of such polyhedral invariants~$I_q$.
  Let $(q_t,\ve{x}_t)_{t \in \bbN}$ be the run of the counter machine.
  Since the run is infinite, there is (at least) one control location $q$ that it visits infinitely often.
  The infinitely many points $(\ve{x}_t,t,\theta(t),t,\theta(t))$ such that $q_t=q$ are all reachable by the transition system and thus belong to~$I_q$;
  call $\ve{A}$ one of these points.
  $I_q$ is thus unbounded in the $(\ve{0},0,0,1,0)$ direction.

  The recession cone of $I_q$, considered as a rational polyhedron, must thus contain a vector $\ve{r}$ such that
  $(\ve{0},0,0,1,0) \cdot \ve{r} > 0$;
  without loss of generality, we can assume it to have integer coordinates.
  Call the last two coordinates of the transition system $t'$ and $y'$.
  For all $\lambda \in \bbN$, $\ve{A}+\lambda \ve{r}$ is also in~$I_q$. 
  Along $\ve{A}+\lambda \ve{r}$, as $\lambda$ increases, $t'$ grows linearly with $\lambda$, thus $\theta(t')$ grows quadratically, yet $y'$ changes linearly, thus the defect $y' - \theta(t')$ goes to $-\infty$. In particular, we can choose $\lambda$ so that the defect is at most $-1$ (note that $t'$ is an integer, since $\ve{A}$ and $\ve{r}$ are integral).
  Since the defect is preserved when applying rewinding transitions, rewinding from such a $\ve{A} + \lambda \ve{r}$ will eventually reach $t'=0$ and $y' \leq -1$, enabling a transition to the bad state, which contradicts the assumption of suitability.
\end{proof}

\begin{theorem}
  The existence of polyhedral inductive invariants suitable for showing that a control location is unreachable is undecidable for programs with integer or rational variables and linear arithmetic, even when simultaneously restricted to deterministic programs, closed transition guards (non-strict inequalities), and $6$ numerical variables.
  The same applies to bounded polyhedral invariants.
\end{theorem}

\begin{proof}
  From Lemmas \ref{lem:closed_termination_implies_invariants} and \ref{lem:nontermination_implies_no_invariants}, and the well-known undecidability of termination of Minsky machines over $\bbZ$~\citep{Minsky67}. The same machine definition also works over~$\bbQ$.
\end{proof}

\begin{remark}
  These problems are semidecidable, and thus, by the above reduction, complete for recursively enumerable problems.
\end{remark}

\begin{proof}
  Enumerate all candidate inductive invariants, stop when one is suitable: the inductiveness check is satisfiability modulo linear (integer or rational) arithmetic.
\end{proof}

\section{Conclusion}
The existence of polyhedral invariants is undecidable for programs over linear arithmetic.
The reduction from counter machines improves on one found in \cite{Monniaux_Acta_Informatica_2018} by doing away with a quadratic guard.
General inference approaches for such invariants are thus necessarily incomplete (heuristic) or restricted to specific subclasses of polyhedra.

Our reduction uses $6$ numerical variables, and as many control locations as necessary to show that termination is undecidable for $2$-counter machines.
We can reduce the number of control locations to $1$ (not counting the bad state) by encoding control into an extra numerical dimension (thus going to~$7$).
The problem remains open for dimensions between $2$ and~$5$ included (dimension $1$ is intervals, thus decidable).

\bibliographystyle{plainnat}
\bibliography{polyhedral_invariants_are_undecidable}

\begin{thebibliography}{9}
\providecommand{\natexlab}[1]{#1}
\providecommand{\url}[1]{\texttt{#1}}
\expandafter\ifx\csname urlstyle\endcsname\relax
  \providecommand{\doi}[1]{doi: #1}\else
  \providecommand{\doi}{doi: \begingroup \urlstyle{rm}\Url}\fi

\bibitem[Cousot and Cousot(1977)]{CousotCousot77}
P.~Cousot and R.~Cousot.
\newblock Abstract interpretation: a unified lattice model for static analysis
  of programs by construction or approximation of fixpoints.
\newblock In \emph{Conference Record of the 4th ACM Symposium on Principles of
  Programming Languages}, pages 238--252, Los Angeles, CA, January 1977.
\newblock \doi{10.1145/512950.512973}.

\bibitem[Cousot(1978)]{Cousot78}
Patrick Cousot.
\newblock \emph{Mé\-tho\-des ité\-ra\-ti\-ves de cons\-truc\-tion et
  d'ap\-pro\-xi\-ma\-tion de points fi\-xes d'opé\-ra\-teurs mo\-no\-to\-nes
  sur un treil\-lis, ana\-ly\-se sém\-an\-ti\-que de pro\-gram\-mes}.
\newblock Thè\-se d'état ès scien\-ces ma\-thé\-ma\-ti\-ques,
  Uni\-ver\-si\-té scien\-ti\-fi\-que et mé\-di\-ca\-le de Gre\-no\-ble,
  Grenoble, France, 21 mars 1978.

\bibitem[Cousot and Halbwachs(1978)]{CousotHalbwachs78}
Patrick Cousot and Nicolas Halbwachs.
\newblock Automatic discovery of linear restraints among variables of a
  program.
\newblock In \emph{Conference Record of the Fifth Annual ACM Symposium on
  Principles of Programming Languages (POPL '78)}, pages 84--96, Tucson, AZ,
  January 1978. ACM Press.
\newblock \doi{10.1145/512760.512770}.

\bibitem[Fijalkow et~al.(2025)Fijalkow, Lefaucheux, Ohlmann, Ouaknine, Pouly,
  and Worrell]{DBLP:journals/jacm/FijalkowLOOPW25}
Nathana{\"{e}}l Fijalkow, Engel Lefaucheux, Pierre Ohlmann, Jo{\"{e}}l
  Ouaknine, Amaury Pouly, and James Worrell.
\newblock On the {M}onniaux problem in abstract interpretation.
\newblock \emph{J. {ACM}}, 72\penalty0 (2):\penalty0 11:1--11:51, 2025.
\newblock \doi{10.1145/3704632}.

\bibitem[Halbwachs(1979)]{Halbwachs_PhD}
Nicolas Halbwachs.
\newblock \emph{Détermination automatique de relations linéaires vérifiées
  par les variables d'un programme}.
\newblock PhD thesis, Université scientifique et médicale de Grenoble \&
  Institut national polytechnique de Grenoble, 1979.
\newblock URL \url{https://theses.hal.science/tel-00288805v1}.

\bibitem[Minsky(1967)]{Minsky67}
Marvin~L. Minsky.
\newblock \emph{Computation: Finite and Infinite Machines}.
\newblock Prentice-Hall, Englewood Cliffs, NJ, 1967.

\bibitem[Monniaux(2019)]{Monniaux_Acta_Informatica_2018}
David Monniaux.
\newblock On the decidability of the existence of polyhedral invariants in
  transition systems.
\newblock \emph{Acta Informatica}, 56\penalty0 (4):\penalty0 385--389, 2019.
\newblock ISSN 0001-5903.
\newblock \doi{10.1007/s00236-018-0324-y}.
\newblock URL \url{https://hal.archives-ouvertes.fr/hal-01587125/}.

\bibitem[Monniaux and Le~Guen(2012)]{DBLP:journals/entcs/MonniauxG12}
David Monniaux and Julien Le~Guen.
\newblock Stratified static analysis based on variable dependencies.
\newblock In Damien Mass{\'{e}} and Laurent Mauborgne, editors,
  \emph{Proceedings of the Third International Workshop on Numerical and
  Symbolic Abstract Domains, NSAD@SAS 2011, Venice, Italy, September 13, 2011},
  volume 288 of \emph{Electronic Notes in Theoretical Computer Science}, pages
  61--74. Elsevier, 2012.
\newblock \doi{10.1016/J.ENTCS.2012.10.008}.

\bibitem[Schrijver(1998)]{Schrijver98}
Alexander Schrijver.
\newblock \emph{Theory of Linear and Integer Programming}.
\newblock Wiley, 1998.
\newblock ISBN 0-471-98232-6.

\end{thebibliography}
\end{document}